\documentclass[a4paper,11pt]{article}

\usepackage{geometry}

\usepackage[utf8]{inputenc}
\usepackage{mathrsfs,amssymb,dsfont,array}
\usepackage{theorem}
\usepackage{verbatim}
\usepackage[intlimits]{empheq}
\usepackage{latexsym}
\usepackage{graphicx}
\usepackage{enumerate}
\usepackage{multirow}
\usepackage{rotating}
\usepackage{subfigure,nicefrac,booktabs}
\usepackage[ruled,vlined]{algorithm2e}
\usepackage[round]{natbib}
\usepackage[%
            font  ={footnotesize},
            labelfont = {bf},
            margin =0.5cm,
            aboveskip = 5pt,
            position = bottom]{caption}

\newcommand{\R}{\mathbb{R}}

\newcommand{\E}{\mathbb{E}}
\newcommand{\N}{\mathbb{N}}

\newcommand{\ccL}{\mathscr{L}}

\newcommand{\BOX}{\mbox{{\ensuremath{\Box}}\hspace{-0.5mm}}}

\makeatletter

\renewcommand{\p@enumi}{\thetheorem-}
\@addtoreset{equation}{section}
\makeatother

\bibpunct[, ]{(}{)}{,}{a}{}{,}

\numberwithin{figure}{section}%
\numberwithin{table}{section}

 \renewcommand\appendix{\par
   \setcounter{section}{0}%
   \setcounter{subsection}{0}%
   \setcounter{figure}{0}%
   \renewcommand\thesection{\Alph{section}}%
   \renewcommand\thefigure{\Alph{section}.\arabic{figure}}}

\newtheorem{theorem}{Theorem}[section]

\newtheorem{proposition}[theorem]{Proposition}

\newtheorem{lemma}[theorem]{Lemma}

\newtheorem{remark}[theorem]{Remark}
\newenvironment{proof}{{\par\noindent\textbf{Proof:}}}{\mbox{}\hfill$\BOX$\\}

\newtheorem{innercustomthm}{Theorem}
\newenvironment{customthm}[1]
  {\renewcommand\theinnercustomthm{#1}\innercustomthm}
  {\endinnercustomthm}

\usepackage[backgroundcolor=white,bordercolor=orange]{todonotes}

\newcommand{\ccF}{{\mathscr F}}

\newcommand{\cA}{\mathcal{A}}

\newcommand{\bbF}{\mathbb{F}}
\newcommand{\ccP}{{\mathscr P}}

\RequirePackage{colortbl}
\definecolor{tscolor}{rgb}{1.0,0.6,0.0}

\usepackage[colorlinks,urlcolor=red,citecolor=blue,linkcolor=red]{hyperref}

\title{Robust deep hedging}

\author{Eva L\"{u}tkebohmert$^{1}$, Thorsten Schmidt$^{2}$, Julian Sester$^{3}${}
}

\begin{document}

\maketitle

\vspace{-6ex}
\begin{center}
\small\textit{$^{1}$Department of Quantitative Finance,\\ Institute for Economic Research, University of Freiburg,\\ Rempartstr. 16,
79098 Freiburg, Germany.\\[2mm]
$^{2}$Department of Mathematical Stochastics, Mathematical Institute,\\
University of Freiburg, Ernst-Zermelo-Str. 1, 79104 Freiburg, Germany.\\[2mm]
$^{3}$NTU Singapore, Division of Mathematical Sciences,\\ 21 Nanyang Link, Singapore 637371}                                                                                                                              \end{center}

\begin{abstract}
We study pricing and hedging under parameter uncertainty for a class of Markov processes which we call  {\it generalized affine processes} and which includes the Black-Scholes model as well as the constant elasticity of variance (CEV) model as special cases. Based on a general  dynamic programming principle, we are able to link the associated nonlinear expectation to a variational form of the Kolmogorov equation which opens the door for fast numerical pricing in the robust framework. 

The main novelty of the paper is that we propose a deep hedging approach which efficiently solves the hedging problem under parameter uncertainty. We numerically evaluate this method on simulated and real data and show that the robust deep hedging outperforms existing hedging approaches, in particular in highly volatile periods. \\[2mm]

\noindent
{\bf Keywords:} 
affine processes, Knightian uncertainty, Kolmogorov equation, deep learning, robust hedging
	
\noindent{{\bf JEL classification:}
C02, C45, G13}

\end{abstract}

\section{Introduction}\label{sec-Intro}

{\it Uncertainty}, as coined by Frank Knight, refers to the case where a number of models (technically: probability measures) are available and one is not able to distinguish between them. This applies for example to the prediction of the evolution of a stock in the future. Even if we have a reliable and rich source of historic information, predicting the future evolution, the future variance or even the whole future distribution is highly complicated. On the one side, this is due to the classical estimation problem: estimated parameters allow for confidence intervals which need to be taken into account for the prediction. On the other side, in particular in financial markets, changes in the underlying dynamics are rather the rule than the exception and additional uncertainty and model risk come into effect, resulting in a widening of confidence intervals.
For pricing, one can efficiently rely on the calibration to option surfaces with all its difficulties. For hedging, when one wants to incorporate the performance under the objective measure and not under the risk-neutral one, this becomes much more challenging. 

Our paper addresses exactly this setting and suggests a deep learning approach for hedging under parameter uncertainty.
The basis for our work is the recently developed class of affine processes under parameter uncertainty, see \cite{FadinaNeufelSchmidt2019}, which we simply call nonlinear affine processes, referring to the associated nonlinear expectation arising from the pricing problem under uncertainty in this class. We extend this approach to those Markovian processes which satisfy
\begin{align}
 dX_t &= (b_0 + b_1 X_t) dt + (a_0 + a_1 X_t)^\gamma dW_t,
\end{align}
where we allow for parameter uncertainty in all the parameters $b_0,b_1,a_0,a_1$, and $\gamma$. We develop the theory for this class of processes which we call nonlinear generalized affine (NGA) processes. The robust pricing problem is solved by utilizing a general dynamic programming principle and establishing the nonlinear Kolmogorov equation, which opens the door for fast (and well-known) numerical approaches. 
In order to solve the hedging problem under parameter uncertainty, we rely on a deep learning approach. To the best of our knowledge, this is the first attempt of this kind.  

We numerically evaluate this method first on simulated data and show that the robust deep hedging outperforms existing hedging approaches when parameter uncertainty is present.
For a realistic data application, we consider the COVID-19 period. 
In this period, stock markets experienced unexpectedly high volatility and variation in the price paths, which poses a huge challenge to classical hedging approaches. When applying robust methods, the first challenge is to find reliable estimates for the parameter intervals specifying the uncertainty in the considered model class. We propose a sliding-window maximum-likelihood estimation approach for this whose maximal and minimal parameter estimates lead to the targeted intervals. With this uncertainty specification at hand, we are able to show that in the considered data examples the robust deep hedging approach leads to a remarkably smaller hedging error in comparison to classical hedging strategies.

Our paper relates to a rich stream of literature motivated by \emph{parameter uncertainty}, dating back to \cite{AvellanedaLevyParas.95,wilmott1998uncertain}, and \cite{FouqueRen2014}. More recent contributions are \cite{cohen2017european,barnett2020pricing, aksamit2020robust, cheridito2017duality, akthari2020generalized}. In the context of option pricing and efficient hedging, \cite{bouchard2015arbitrage,hou2018robust} developed approaches respecting an ambiguity set of possible underlying probability measures %
and \cite{acciaio2016model, beiglbock2013model, cox2011robust, dolinsky2014martingale, hobson1998robust,lutkebohmert2019tightening,nadtochiy2017robust, neufeld2021model} introduced approaches to entirely model-free option pricing and to model-free super-replication. 

Further, our paper contributes to the recent literature on deep learning approaches in hedging, starting from the seminal work \cite{buehler2019deep} and followed by \cite{gumbel2020machine,  cuchiero2020generative, cao2021deep, carbonneau2021deep, carbonneau2021equal, chen2021deep, eckstein2021robust, gierjatowicz2020robust, horvath2021deep, neufeld2021deep}, amongst many others (see also \cite{ruf2020neural} for a review).

The remainder of the paper is organized as follows: In Section \ref{sec-GeneralizedAffine} we introduce the theoretical basis for NGA processes. In Section \ref{sec:robust hedging} we introduce the robust hedging approach, illustrated with simulated examples, while in Section \ref{sec:data} we apply robust hedging to real data. Section \ref{sec:conclusions} concludes and the appendix contains  some proofs. \\

\section{Generalized affine processes under parameter uncertainty}\label{sec-GeneralizedAffine}

In this section we extend the notion of  affine diffusions under parameter uncertainty to a more general setting. To this end, consider a state space $E$ which is either $\R$ or $\R_{>0}$. We start with the setting without parameter uncertainty.

A {\it generalized affine diffusion} is a continuous semimartingale $X$ which is a unique strong solution of the stochastic differential equation (SDE)
\begin{align}\label{equ generalized affine process}
 dX_t &= (b_0 + b_1 X_t) dt + (a_0 + a_1 X_t)^\gamma dW_t, 
\end{align}
with suitably chosen $b_i,a_i\in\mathbb{R}$, $i=0,1$, $\gamma\in[1/2,1]$ and initial value $X_0=x\in E$. Here, $W$ denotes a standard Brownian motion. If we choose $\gamma=1/2$ we obtain the well-known special case of a (continuous) affine process.

In Proposition \ref{prop:existence} in the appendix we utilize the classical existence result of Engelbert and Schmidt (\cite{engelbert1985one,engelbert1985solutions}) to show that a generalized affine diffusion  exists on  a proper state space. 

Fix a time horizon $T>0$ and consider $\Omega=C([0,T])$ as the canonical space of continuous one-dimensional paths.  Denote by $\ccF$ the Borel-$\sigma$-algebra on $\Omega$. Let $X$ be the canonical process $X_t(\omega)=\omega_t$ for $\omega\in\Omega$ and $t\in[0,T]$ and denote by $\bbF=(\ccF_t)_{t\in[0,T]}$ the filtration generated by $X$. 

Let $\ccP(\Omega)$ be the set of all probability measures on $(\Omega,\ccF)$. A probability measure $P\in\ccP(\Omega)$ is called a \emph{semimartingale law} for the process $X$ if there exists a process $B^P$ with continuous paths of (locally) finite variation $P$-a.s. and a continuous local $P$-martingale $M^P$ with $B_0^P=M_0^P=0$ such that $X=X_0+B^P+M^P.$ 
Intuitively, this describes the setting when $X$ is a continuous semimartingale under $P$ which is given as a sum of the integrated drift process $B^P$ and a local martingale $M^P$. 

A continuous semimartingale $X=X_0+B^P+M^P$ is said to admit \emph{absolutely continuous characteristics }$(B^P,{C})$ with ${C}=\langle M^P\rangle$ if there exist predictable processes $\beta^P$ and $\alpha>0$ such that\footnote{Note that $\alpha$ does not depend on $P$ as the quadratic variation is a path property. We therefore write $C^P=C$ and $\alpha^P=\alpha$.}
\begin{align}\label{eq:sem mart char}
B^P=\int_0^{\cdot} \beta_s^Pds,\quad {C}=\int_0^{\cdot}\alpha_sds.
\end{align}

In the case that a continuous semimartingale possesses absolutely continuous characteristics, we can directly consider the drift (instead of the integrated drift).

\subsection{Parameter uncertainty}
Next, we introduce parameter uncertainty in the spirit of Frank Knight. Recall that a generalized affine diffusion is characterized by the five parameters $b_0,b_1,a_0,a_1$ and $\gamma$. The targeted uncertainty we are interested in can be described as follows: instead of assuming the parameter $\theta$ to be known exactly, we introduce an interval $[\underline{\theta},\bar \theta]$ and consider each value in the interval equally likely.
Taking into account this parameter uncertainty leads to a nonlinear setting, which we introduce now.

Denote the considered parameter intervals by $[\underline b_i,\bar b_i]$ and $[\underline a_i,\bar a_i]$ with $i=0,1$, and by $[\underline \gamma, \bar \gamma]$. 
Denote by $\Theta:=[\underline b_0,\bar b_0]\times [\underline b_1,\bar b_1]\times[\underline a_0,\bar a_0]\times [\underline a_1,\bar a_1] \times [\underline \gamma, \bar \gamma]$ the parameter set. 

To transport this parameter uncertainty to stochastic processes we have to be more careful. In the Markovian setting we consider here, the evolution of the process may depend on the current state $x$ of the process. In this regard, we introduce the associated intervals 
\begin{align}\label{def:b(x),a(x)}
b(x) :=  \{ b_0 + b_1 x: b_0,b_1 \in \Theta\},\quad
a(x) :=  \{ (a_0 + a_1 x^+)^{2\gamma} : a_0,a_1,\gamma \in \Theta\}
\end{align}
for $x\in\mathbb{R}$, where $(\cdot)^+:=\max\{\cdot,0\}$ which describe the possible diffusive behaviour of the process $X$ given it is in state $x$. 

\begin{remark}[On the role of the state space]
In the classical one-dimensional affine setting, the state space already defines if the affine process is of Cox-Ingersoll-Ross type (when the state space is $\R_{> 0}$ or $\R_{\ge 0}$) or of Vasi\v cek-type. This is no longer the case when parameter uncertainty is introduced. Indeed, in \cite{FadinaNeufelSchmidt2019}, the non-linear Vasi\v cek-CIR model was introduced which has as state space $\R$ and intuitively is able to capture both sorts of dynamics. In interest rate markets, where negative rates have been neglected for a long time, such an approach efficiently avoids the model risk by the necessity to choose between the state space $\R$ and $\R_{\ge 0}$. In the non-linear setting considered here, this leads to the use of $x^+$ in the definition of $a(x)$. This  ensures non-negativity of the quadratic variation when $E=\R$ but does not restrict unnecessarily the dynamics of the generalized affine process.
\end{remark}

The description of the generalized affine process under parameter uncertainty is now intuitively given by all those probability laws describing a diffusion where the drift and the volatility always stay in the intervals $b(x)$ and $a(x)$ considered at $x=X_s(\omega)$. This means, that we consider all continuous semimartingales whose characteristics stay in the parameter uncertainty bounds. 

More precisely, we introduce the following notion\footnote{The name (nonlinear) generalized affine process and the setting of a NGA is inspired by the master thesis \cite{Denk}.}: a \emph{nonlinear generalized affine process (NGA)} starting in $x\in E$ at time $t\in[0,T]$ is the family of all absolutely continuous semimartingale laws $\cA(t,x,\Theta)$, such that for each $P \in \cA(t,x,\Theta)$ giving rise to the differential characteristics $(\beta^P,\alpha)$ we have
\begin{align} \label{eq:intervals GNLA}
  \beta_s^P \in b(X_s), \qquad \alpha_s \in a(X_s) 
\end{align}
$dt \otimes dP$-almost surely on $(t,T]\times \Omega$ and $P(X_t=x)=1$. We call $P$ \emph{generalized affine dominated} by $\Theta$ on $(t,T]$  or simply \emph{GA-dominated} by $\Theta$. Note that non-negativity of the quadratic variation is ensured by using $(\cdot)^+$ in the definition of $a$ in Equation \eqref{def:b(x),a(x)}

\subsection{Robust pricing of derivatives}

In order to study derivative prices under parameter uncertainty, we consider an European claim with maturity $T$ and  payoff $\psi(X_T)$. Here, $\psi$ is an integrable function $\psi:E\rightarrow\mathbb{R}$. 

Since in our robust setting we do not have a single probability measure at hand, but a family of measures which we treat equally likely, a natural candidate for pricing is the worst-case price: the price which dominates all prices computed under the probability measures we consider. 

In this regard, define the \emph{value function} $v:[0,T]\times E\rightarrow \mathbb{R}$ by
$$
v(t,x):=\sup_{P\in \cA(t,x,\Theta)} \mathbb{E}^P[\psi(X_T)].
$$
Analogously one can define a lower bound of possible prices 
$\inf_{P\in \cA(t,x,\Theta)} \mathbb{E}^P[\psi(X_T)]$, which, due to the relation $\inf_{P\in \cA(t,x,\Theta)} \mathbb{E}^P[\psi(X_T)]=-\sup_{P\in \cA(t,x,\Theta)} \mathbb{E}^P[-\psi(X_T)]$, can be studied with identical methods. We therefore focus on the upper bound.

A central tool for establishing a nonlinear Kolmogorov equation and therefore the tractability of the setting is the \emph{dynamic programming principle}. Intuitively it states that if we consider a stopping time between $t$ and $T$ and compute the price (the value function) at that stopping time and take expectations of this random quantity, we obtain the value function at time $t$. Thus, the value can not be improved by however skilled stopping. 

\begin{proposition}\label{prop DPP}
Consider a nonlinear generalized affine process with state space $E$ and a stopping time $\tau$ on $[t,T].$ For any $(t,x)\in[0,T]\times \Omega$, we have
\begin{align}\label{eq:DPP}
v(t,x)=\sup_{P\in \cA(t,x,\Theta)} \mathbb{E}^P[v(\tau,X_{\tau})].
\end{align}
\end{proposition}

The proof of this result follows similarly to the proof of the dynamic programming principle in \cite{FadinaNeufelSchmidt2019} which is based on Theorem 2.1 in \cite{ElKarouiTan2013}. The necessary measurability and stability conditions which are proved for the affine case in Lemma 1 and 2 of \cite{FadinaNeufelSchmidt2019} can be shown similarly for the generalized affine case.

\subsection{The nonlinear Kolmogorov equation}

Note that the computation of the value function according to Equation \eqref{eq:DPP}, or the robust upper price of a derivative is not as easily accessible  by Monte-Carlo estimation as in the classical case. Indeed, as is clear form Equation \eqref{eq:DPP}, it is not sufficient to simulate different paths under different distributions but we need to obtain Monte-Carlo estimates of expectations with a fixed probability measure $P$ and then need to find the supremum of these expectations. If no monotonicity can be exploited, this will be difficult to compute.

However, a very efficient tool can be developed which is a nonlinear version of the Kolmogorov equations. By relying on numerical methods for nonlinear partial differential equations, we will be able to compute the value function within seconds. 

A central tool for describing Markov processes is the infinitesimal generator. For the generalized affine process $X$ (under no parameter uncertainty, see \eqref{equ generalized affine process}), the infinitesimal generator is given by 
$$
\ccL^{\theta}f(x)=(b_0+b_1 x)\partial_x f(x)+\frac{1}{2}  (a_0+a_1 x^+)^{2\gamma} \partial_{xx} f(x),
$$
with $f\in C^2(\R)$. 

For the nonlinear version of the Kolmogorov equation we will take the worst-case generator, i.e.~the supremum over all generators with $\theta \in \Theta$. 
More precisely, for some integrable $\psi:E\rightarrow \mathbb{R}$ consider the nonlinear partial differential equation
\begin{equation}\label{equ nonlinear PDE}
\left\{
\begin{array}{rcccl}
\partial_t u +G(x,\partial_x u(t,x),\partial_{xx} u(t,x))&=&0 \quad &\mbox{on}&\, [0,T)\times E\\
u(T,x)&=& \psi(x)\quad &\mbox{for}& \, x\in E,
\end{array}\right.
\end{equation}
where $G:E\times \mathbb{R}\times \mathbb{R}\rightarrow \mathbb{R}$ is defined via
\begin{align}\label{eq:def:PDE}
G(x,p,q):=\sup_{(b_0,b_1,a_0,a_1,\gamma)\in\Theta} \left\{(b_0+b_1 x)p+\frac{1}{2}(a_0+a_1 x^+)^{2\gamma} q\right\}.
	\end{align}

We then obtain the value function as viscosity solution of the nonlinear Kolmogorov equation. 

\begin{theorem}\label{thm:Kolmogorov}
Consider a family of nonlinear generalized affine processes with state space $E$ and a Lipschitz continuous payoff function $\psi:E\rightarrow\mathbb{R}$. Then,
$$
v(t,x):=\sup_{P\in \cA(t,x,\Theta)}\mathbb{E}^P[\psi(X_T)],\quad \mbox{for } x\in E, t\in[0,T],
$$
is a viscosity solution to the PDE (\ref{equ nonlinear PDE}).
\end{theorem}

The result follows by the similar arguments as in the proof of Theorem 1 in \cite{FadinaNeufelSchmidt2019}. We  relegate the proof to the appendix.

\bigskip

\section{Robust hedging}\label{sec:robust hedging}

After the setting for NGA processes has been detailed and the pricing discussed, we come to the main novelty of the paper: the efficient computation of hedging strategies. It is our goal to also find a numerical procedure which replaces the classical Monte-Carlo estimation in a robust setting. Motivated by Theorem \ref{thm:Kolmogorov}, we will proceed as follows: first, we discretize in time and utilize the Euler Maruyama approximation of the generalized affine process as given in Equation \eqref{equ generalized affine process}. Second, at each time step, we select a new parameter set $\theta \in \Theta$ by sampling from a uniform distribution. Note that sampling from a uniform distribution corresponds to assigning equal weight to all probability measures under consideration, which seems adequate for a robust hedging approach that can also be applied in situations which are underestimated by methods solely relying on historical data, see also Remark~\ref{rem_bayesian}, where we discuss possible extensions.

 These processes serve as an approximation of the class $\cA(t,x,\Theta)$. In our robust deep hedging approach we train our network on these samples and determine the hedging function which minimizes the hedging error over all these samples. 

Since this section is mainly on numerics, we take the freedom to generalize the setup of Section \ref{sec-GeneralizedAffine} slightly by also allowing for path-dependent derivatives. The theoretical subtleties which build the basis for this step can be found in \cite{GeuchenSchmidt2021}.

A path-dependent derivative allows that the payoff at maturity $T$ depends on the full path of the process $X$ up to time $t$, $(X_t)_{0 \le t \le T}$. We denote the square-integrable payoff by $\Phi_T:=\Phi((X_t)_{0\leq t \leq T})$, with a measurable  function $\Phi:C([0,T])\rightarrow \R_{>0}$. Our aim is to determine hedging strategies $(h_t)_{0\leq t \leq T}$ and cash positions $d\in \R$ such that the quadratic error is minimized
\begin{equation}\label{eq_minimize_quadratic}
\min_{(h_t)_{0\leq t \leq T},d\in \R}\E^P \bigg[\Big(d+\int_0^T h_t ~d X_t - \Phi_T\Big)^2\bigg]
\end{equation}
for all $P\in \cA(0,x_0,\Theta)$. This formulation is a consequence of the considered model ambiguity, under which every measure from $\cA(0,x_0,\Theta)$ is taken into account.
In the following we develop a deep learning approach to compute the hedging strategy $h$.

\subsection{A numerical procedure relying on deep neural networks}

First, we discretize the interval $[0,T]$ through $0 =t_0 \leq t_1 \leq \cdots \leq t_n=T$.
Next, we approximate the hedging strategy $h_{t_i}$ at grid point $t_i$ through neural networks. We start with a precise definition of neural networks, referring to \cite{petersen2020neural} for a detailed mathematical study on this topic.

Let $\varphi:\R\rightarrow \R$ be a non-constant function, called \emph{activation function}.
A (feed-forward) neural network with input dimension $d_{\operatorname{in}}\in \N$, output dimension $d_{\operatorname{out}}\in \N$, $l\in \N$ layers, and activation function $\varphi$ is a function of the form 
\begin{align*}
\R^{d_{\operatorname{in}}}&\rightarrow\R^{d_{\operatorname{out}}}\\
x &\mapsto A_l \circ \varphi \circ A_{l-1} \circ  \dots \circ \varphi\circ A_0(x),
\end{align*}
where $(A_i)_{i=0,\dots,l}$ are affine functions $A_i:\R^{h_i} \rightarrow \R^{h_{i+1}}$, and where the activation function is applied component-wise. The number $h_i \in \N$ is called the \emph{number of neurons} of layer~$i$.
We say a neural network is \emph{deep} if $l \geq 2$, and we denote the class of all neural networks with input dimension $d_{\operatorname{in}}$, output dimension $d_{\operatorname{out}}$, $l$ layers and activation function $\varphi$ by $\mathcal{NN}^{l,\varphi}_{d_{\operatorname{in}},d_{\operatorname{out}}}$.

To solve the minimization problem stated in \eqref{eq_minimize_quadratic} for arbitrary measures $P\in \cA(0,x_0,\Theta)$ we sample paths of  $(X_{t_i})_{0 \leq i \leq n}$, where we sample each path under a newly randomly picked measure $P\in \cA(0,x_0,\Theta)$, where the parameters are uniformly chosen from $\Theta$ in each time step. We then compute the quadratic hedging error on a batch of samples and optimize the neural network to minimize the quadratic hedging error. This procedure is summarized in Algorithm~\ref{algo_nn} and builds on the findings from \cite{buehler2019deep}, where however no ambiguity w.r.t.\,the choice of the correct underlying probability measure is taken into account.

\begin{algorithm}[t!]
\SetAlgoLined
\SetKwInOut{Input}{Input}
\SetKwInOut{Output}{Output}

\ \\

\Input{parameter set $\Theta$; hyperparmeters of the neural network such as number of layers $l \in \N$, number of neurons, activation function $\varphi$, learning rate of the optimizer; number of iterations $N_{\operatorname{iter}}$; batch size $B$; payoff function $\Phi((X_t)_{0\leq t \leq T})$; discretization $0 =t_0 \leq t_1 \leq \cdots \leq t_n=T$; initial value $x_0$;}
\Output{parameter $d$: cash position of the hedging strategy;\\ neural network $h\in \mathcal{NN}^{l,\varphi}_{2,1}$: self-financing strategy, inputs $t$ and $X_{t}$;}
\ \\[-2mm] 

Initalize the parameters of the neural network $h\in \mathcal{NN}^{l,\varphi}_{2,1}$ randomly;\\[0.5mm]
Initalize parameter $d =0$;\\[0.5mm]
\For{$\operatorname{iter} =1,\dots,N_{\operatorname{iter}}$}{
\For{$b=1,\dots,B$}{

\ \\

Generate paths of the generalized affine process using the Euler-Maruyama method:\\[0.5mm]
$X_0^b:= x_0$, $\Delta t_i: = t_{i+1}-t_{i}$\\[0.5mm]

\For{$i=0,\dots,n-1$}{ \ \\
Generate $\Delta W_i \sim N(0,\Delta t_i)$;\\[1mm]
Generate $\gamma^{(i)}\sim U\left([\underline{\gamma},\overline{\gamma}]\right)$, $a_0^{(i)} \sim U\left([\underline{a_0},\overline{a_0}]\right)$, $a_1^{(i)} \sim U\left([\underline{a_1},\overline{a_1}]\right)$, \\
    \ \ \ \ \ \ \ \ \ \ \ \ \ $b_0^{(i)} \sim U\left([\underline{b_0},\overline{b_0}]\right)$, $b_1^{(i)} \sim U\left([\underline{b_1},\overline{b_1}]\right)$;\\[1mm]
set $X^b_{i+1} := X_i^b+(b_0^{(i)}+b_1^{(i)} X_i^b) \Delta t_i +(a_0^{(i)}+a_1^{(i)} X_i^+)^{\gamma^{(i)}}  \Delta W_i$
}
}
Apply stochastic gradient descent / backpropagation  to minimize the loss
\[
\sum_{b=1}^B\left( d + \sum_{i=0}^{n-1} h(t_i,X_i^b)(X_{i+1}^b-X_i^b)-\Phi\left((X_i^b)_{i=1,\dots,n}\right)\right)^2
\] w.r.t.\,the parameters of $h$ and w.r.t.\ $d$
}

 \caption{Computation of Optimal Hedging Strategies}\label{algo_nn}
\end{algorithm}

\begin{remark}[Training Time]
The main source of difference w.r.t.\,computational time of Algorithm~\ref{algo_nn} in comparison with the deep hedging approach from \cite{buehler2019deep} turns out to be the random sampling of the $5$ parameters while creating the paths on which we train our neural networks. This sampling step, which reflects the parameter uncertainty in our approach, is not necessary when applying the approach from \cite{buehler2019deep}. 
We found however that the speed difference in practice is not very pronounced.
In  the setting of Section~\ref{sec_call_option} and with the neural network architecture as specified in the beginning of Section~\ref{sec_numerical} we tested that the approach of \cite{buehler2019deep} runs approximately $1.24$ times faster on a standard computer. ($396$ seconds vs. $318$ seconds for $1{,}000$ iterations of training).
\end{remark}

\begin{remark}[Uniform distribution of parameters]\label{rem:prior distributions}
For training the hedging strategy in Algorithm \ref{algo_nn} we simulate samples from a nonlinear generalized affine process in the spirit of an Euler-Maruyama scheme: in each time step $i$ we simulate the Euler-Maruyama discretization and  choose parameters according to Equation \eqref{eq:intervals GNLA}, i.e.~we draw the parameters $\gamma, a_0,a_1,b_0,b_1$ uniformly from $\Theta$.
This seems to be the natural choice for a robust setting since one is interested in putting equal weights on all possible scenarios.
An alternative, but typically more costly strategy, would be to choose a certain discretization of $\Theta$ and to consider \emph{all} gridpoints in each step. Another alternative would be a Bayesian a posteriori distribution, as noted in Remark \ref{rem_bayesian}. 
\end{remark}

\begin{remark}[The quadratic loss function]
In Algorithm \ref{algo_nn} we chose a quadratic loss function which balances gains and losses from the seller and buyer symetrically and therefore leads to a fair hedging price which seems reasonable in many practical applications. If, however, one is rather interested in a classical robust hedging which dominates all hedging strategies for each $P \in \cA(0,x_0,\Theta)$, one would choose a loss function which penalizes losses but not gains as for example a risk measure. This procedure was also suggested in \cite{buehler2019deep}, but is not studied further here. 
\end{remark}

\subsection{Numerical Experiments}\label{sec_numerical}
We apply the presented numerical routine from Algorithm~\ref{algo_nn} in several examples. For all of the examples in this section we consider a nonlinear generalized affine process with parameters specified through
 \begin{equation}\label{eq_parameters_nlap}
\begin{aligned}
x_0 & = 10\\
a_0 &\in [0.3,0.7],~~~ a_1 \in [0.4,0.6], \\
b_0 &\in [-0.2,0.2],~~~ b_1 \in [-0.1,0.1], \\
\gamma & \in [0.5,1.5].
\end{aligned}
\end{equation}

To train neural networks  $h\in \mathcal{NN}^{l,\varphi}_{2,1}$ according to Algorithm~\ref{algo_nn} we specify their architecture as follows. We apply to all layers of the neural network the \emph{ReLU}-activation function $\varphi(x)=\max\{x,0\}$ while the neural networks possess $l=4$ layers with $256$ neurons each. Moreover, Algorithm~\ref{algo_nn} is implemented using the \emph{Tensorflow}-environment~(\cite{abadi2016tensorflow}), in which we execute Algorithm~\ref{algo_nn} with a batch size of $256$ and by employing the Adam optimizer (\cite{kingma2014adam}) with standard parameters and a learning rate of $0.005$ for the backpropagation step.
The used \emph{Python}-codes are provided for convenience and can be found under \href{https://github.com/juliansester/nga}{https://github.com/juliansester/nga}.

In the following we will evaluate the performance based on the relative hedging error, defined as hedge minus payoff of the derivative divided by the respective price of the hedging strategy. Compared to the (relative) quadratic hedging error this has the advantage that we also observe the direction of the error.

\subsubsection{Hedging of at-the-money call options}\label{sec_call_option}
First, we compute, by applying Algorithm~\ref{algo_nn}, an optimal hedging strategy for fixed parameters $a_0=0.5$, $a_1=0.5$, $b_0=0$, $b_1=0$, $\gamma =1$ (the mean of each of the respective intervals in \eqref{eq_parameters_nlap}) for a call option with payoff $\Phi_T=(X_T-x_0)^+$ for $T=30/365$.

Then, for the same derivative, we consider parameter uncertainty by taking into account uncertainty w.r.t.\ the model parameters as specified in \eqref{eq_parameters_nlap}. By applying Algorithm~\ref{algo_nn} with $n=30$, we determine the optimal hedging strategy under parameter uncertainty.
In the left panel of Figure~\ref{fig_strat_gamma_certain_call}, we depict the optimal hedging strategy with fixed parameters, obtained by Algorithm~\ref{algo_nn} after $10{,}000$ iterations, whereas in the right panel of  Figure~\ref{fig_strat_gamma_certain_call}, we depict the optimal hedging strategy, computed with the same number of iterations, when including parameter uncertainty.

\begin{figure}[t!]
\begin{center}
\includegraphics[width=0.45\textwidth]{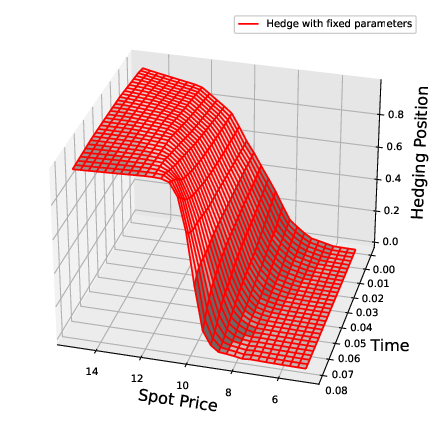}
\includegraphics[width=0.45\textwidth]{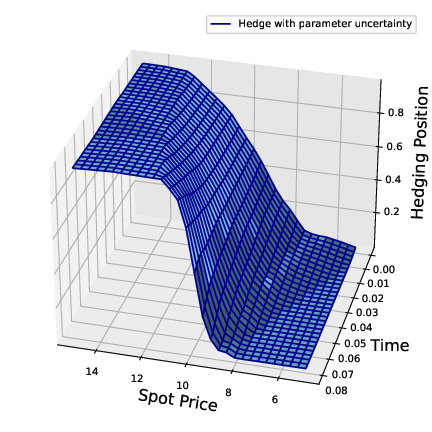}
\end{center}
\caption{ Hedging of at-the-money call options. \underline{Left:} Hedging strategy for fixed parameters $a_0=0.5$, $a_1=0.5$, $b_0=0$, $b_1=0$, $\gamma =1$. \underline{Right:} Robust hedging taking into account parameter uncertainty as specified in \eqref{eq_parameters_nlap}. }\label{fig_strat_gamma_certain_call}
\end{figure}

\begin{figure}[t!]
\begin{center}
\includegraphics[width=0.8\textwidth]{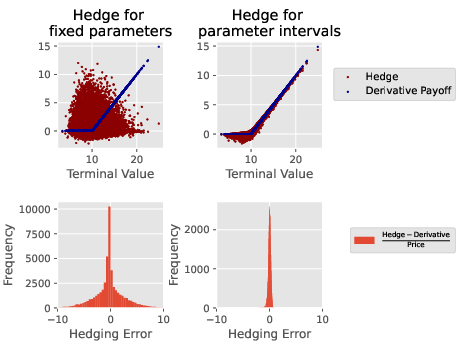}

\caption{Hedging of an at-the-money call option. \underline{Top:} The left panel shows the optimal hedge for fixed parameters $a_0=0.5$, $a_1=0.5$, $b_0=0$, $b_1=0$, $\gamma =1$ and the right panel shows the hedge under parameter uncertainty with parameter intervals as in \eqref{eq_parameters_nlap}, both 
evaluated on $50{,}000$ paths generated according to \eqref{eq_parameters_nlap}. 
\underline{Bottom:}  The figures depict the relative hedging error of the hedge, trained with fixed parameters (left panel) and trained under parameter uncertainty (right panel) with parameter intervals as in \eqref{eq_parameters_nlap}. }
\label{fig_strat_gamma_certain_call_hedging}
\end{center}
\end{figure}

Even though both strategies look very similar at first sight, they perform very differently on random paths generated under parameter uncertainty. For an illustration of this effect, we compute the relative hedging error of both strategies on $50{,}000$ paths that are generated according to the parameters from \eqref{eq_parameters_nlap}, i.e., under parameter uncertainty. The results are displayed in  Figure~\ref{fig_strat_gamma_certain_call_hedging} and reveal that the hedging strategy which was trained on paths that take into account parameter uncertainty possesses a remarkably smaller hedging error in comparison with the strategy which was trained on paths with fixed parameters and which is optimal for these. In Table~\ref{tbl_hedges_robust_fix} we provide mean and standard deviation of the relative hedging errors  verifying the observation that the robust hedging strategy outperforms in this scenario the non-robust hedging strategy.

\subsubsection{Hedging of a butterfly option}\label{sec_butterfly}
While a call option has a high degree of monotonicity, we now explore a more complicated option, a butterfly option. Note that in classical \emph{linear} pricing, one can obtain  the price of a butterfly as the sum of the prices of calls and puts. This is no longer true in the nonlinear case, the case with parameter uncertainty, since the supremum destroys the linearity. Thus, nonlinear pricing in this setting is substantially more involved.

We consider a NGA process with  parameters as in Equation \eqref{eq_parameters_nlap} and a \emph{butterfly} payoff function given by 
\[
\Phi_T=(X_T-8)^++(X_T-12)^+-2\cdot(X_T-10)^+.
\]

We depict the optimal hedging strategy which was computed according to Algorithm~\ref{algo_nn} with $n=30$ in the left panel of Figure~\ref{fig_butterfly_hedge}. The relative hedging error evaluated on $50{,}000$ samples, created under uncertainty, is illustrated in the middle panel while in the right panel we provide a histogram of the difference between the absolute value of the hedging error of a hedging strategy trained on paths with fixed parameters $a_0=0.5$, $a_1=0.5$, $b_0=0$, $b_1=0$, $\gamma =1$ with the absolute value of the relative hedging error of the robust strategy. The histogram shows that in most of the samples the relative hedging error of the non-robust hedge is larger than the relative hedging error of the robust hedge, an observation which is also verified through Table~\ref{tbl_hedges_robust_fix}.

\begin{figure}[t!]
\begin{center}
\includegraphics[width=0.32\textwidth]{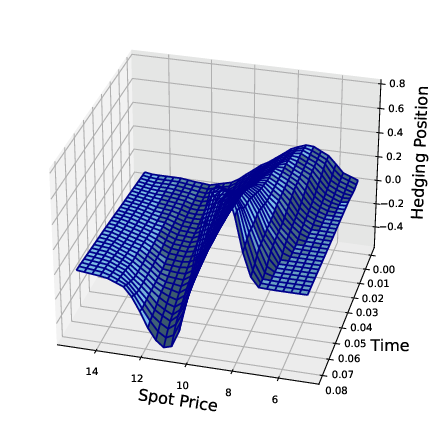}
\includegraphics[width=0.32\textwidth]{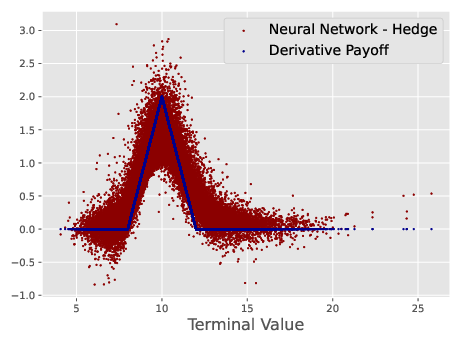}
\includegraphics[width=0.32\textwidth]{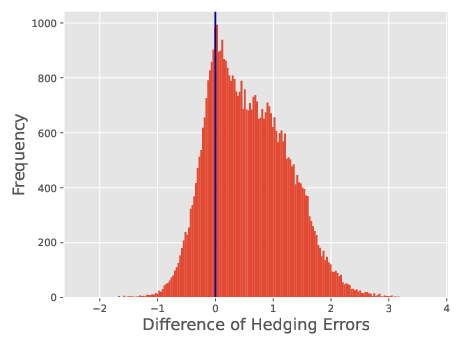}
\end{center}

\caption{\underline{Left:} robust hedging strategy for butterfly option $(X_T-8)^++(X_T-12)^+-2\cdot(X_T-10)^+$. \underline{Centre:} trained robust hedging strategy evaluated on $50{,}000$ samples created under uncertainty. \underline{Right:} histogram of the  difference between the relative hedging error of a hedging strategy with fixed parameters  $a_0=0.5$, $a_1=0.5$, $b_0=0$, $b_1=0$, $\gamma =1$ and the relative hedging error of a robust hedging strategy, evaluated on $50{,}000$ samples created under uncertainty.}\label{fig_butterfly_hedge}

\end{figure}

\subsubsection{Hedging of path-dependent options}\label{sec_path_dependent}
Next, we consider a NGA process with  parameters as in \eqref{eq_parameters_nlap} and a  lookback call option with payoff function
\[
\Phi_T= \left(\max((X)_{0\leq t \leq T})-12\right)^+
\]

We depict the optimal hedging strategy computed with Algorithm~\ref{algo_nn} with $n=30$ in the left panel of Figure~\ref{fig_path_hedge} and the hedging error evaluated on $50{,}000$ samples, created according to uncertainty as in \eqref{eq_parameters_nlap}, in the middle panel of Figure~\ref{fig_path_hedge}. In the right panel of Figure~\ref{fig_path_hedge} we compare the hedging error of a trained non-robust strategy with the hedging error of the trained robust strategy. The trained robust strategy outperforms the non-robust strategy clearly on scenarios that were created under uncertainty according to \eqref{eq_parameters_nlap}, which also can be seen in Table~\ref{tbl_hedges_robust_fix}.

As the payoff function is path-dependent, we could improve the hedging performance further by allowing the self-financing hedging strategy $h_t(X_t,\max_{0\leq s \leq t}X_s)$ to be dependent also on the running maximum. We observe that this approach indeed additionally improves the hedging performance to some degree. The results are displayed in the rightmost column of Table~\ref{tbl_hedges_robust_fix}.
\begin{figure}[t!]
\begin{center}
\includegraphics[width=0.32\textwidth]{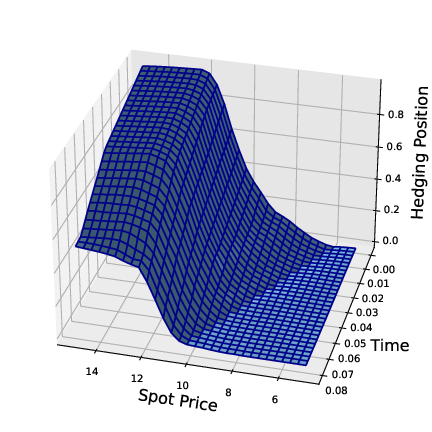}
\includegraphics[width=0.32\textwidth]{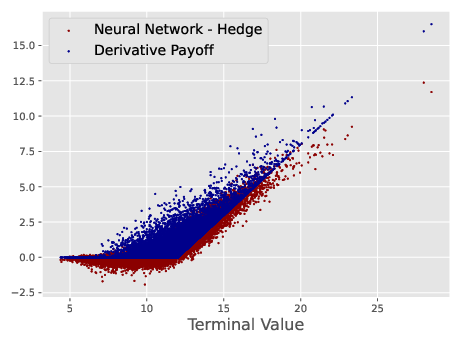}
\includegraphics[width=0.32\textwidth]{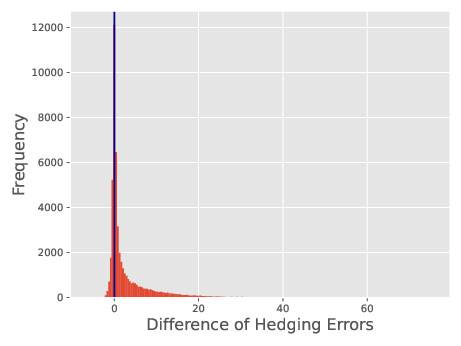}
\end{center}
\caption{\underline{Left:} robust hedging strategy for the lookback call  $\max((X)_{0\leq t \leq T}-12)^+$.
\underline{Centre:} trained robust hedging strategy evaluated on $50{,}000$ samples created under parameter uncertainty  \eqref{eq_parameters_nlap}.
\underline{Right:} histogram of the difference between the absolute value of the hedging error of a strategy trained with fix parameters   $a_0=0.5$, $a_1=0.5$, $b_0=0$, $b_1=0$, $\gamma =1$ and of the absolute value of the hedging error of the robust hedging strategy.}\label{fig_path_hedge}
\end{figure}

\begin{table}[t!]
\begin{center}
\begin{tabular}{llllllll}
\toprule
&\multicolumn{2}{c}{Call  } &\multicolumn{2}{c}{Butterfly   }&\multicolumn{3}{c}{Lookback }\\
\midrule 
         Parameters& fixed & robust & fixed & robust & fixed & robust & run max \\
\midrule

mean &2.0266  & \textbf{0.1924} &0.9078  & \textbf{0.3328} & 3.5582  & 0.5857  & \textbf{0.5476}  \\  
std.\ dev. &2.2650  &0.1636   &0.6548   &0.2648  & 6.0482 & 0.7228   & 0.6305 \\ \bottomrule\end{tabular} 
\end{center}
\caption{\emph{Mean} and \emph{standard deviation} of  the absolute value of the hedging error of different hedging strategies. The strategies were trained according to Algorithm~\ref{algo_nn} for the payoff functions that were discussed in Section~\ref{sec_call_option}, Section~\ref{sec_butterfly} and Section~\ref{sec_path_dependent}, evaluated on $50{,}000$ sample paths created under parameter uncertainty. \emph{Fixed}  hedging strategies were trained for the \emph{fixed} parameters $a_0=0.5$, $a_1=0.5$, $b_0=0$, $b_1=0$, $\gamma =1$ and  \emph{robust} strategies  were trained with parameter uncertainty as in \eqref{eq_parameters_nlap}. 
For the lookback-option we also consider a hedging strategy depending on the running maximum (\emph{run max}), which outperforms the Markovian strategy in the path-dependent case. 
\label{tbl_hedges_robust_fix}}
\end{table}

\subsection{Prices of hedging strategies}\label{sec:hedging_prices}

The presented robust hedging approach allows to respect the problem that in many situations robust price bounds such as $\sup_{P\in \cA(0,x,\Theta)}\mathbb{E}^P[\Phi_T(X_T)]$ are too expensive to have practical relevance (compare e.g. \cite{frey1999bounds}, \cite{biagini2004super} and \cite{neufeld2018buy}).

Robust price bounds may of course still be of interest, for instance to check the market for mispriced derivatives or to compute price bounds when the parameter set $\Theta$ is chosen sufficiently small such that this approach leads to meaningful prices. To compute the price bound $\sup_{P\in \cA(0,x,\Theta)}\mathbb{E}^P[\Phi_T((X_T)]$ one may then solve the corresponding PDE \eqref{equ nonlinear PDE} by using an explicit finite-difference method, compare also \cite{FadinaNeufelSchmidt2019}, where a similar approach in a nonlinear affine setting is pursued and see the companion code on \href{https://github.com/juliansester/nga}{https://github.com/juliansester/nga} for more details.

In Figure~\ref{fig_comparison_prices} we provide the prices of hedging strategies for a call option with strike $K=10$ (as in Section~\ref{sec_call_option}) and of a butterfly option as in Section~\ref{sec_butterfly}, where we consider the parameters as specified in \eqref{eq_parameters_nlap}. For comparison, we also show price bounds $\sup_{P\in \cA(0,x,\Theta)}\mathbb{E}^P[\Phi_T((X_T)]$ and $\inf_{P\in \cA(0,x,\Theta)}\mathbb{E}^P[\Phi_T((X_T)]$ computed with the mentioned explicit finite-difference method. We display prices and price bounds for different initial values of the underlying process. 

\begin{figure}[t!]
\begin{center}
\includegraphics[width=0.45\textwidth]{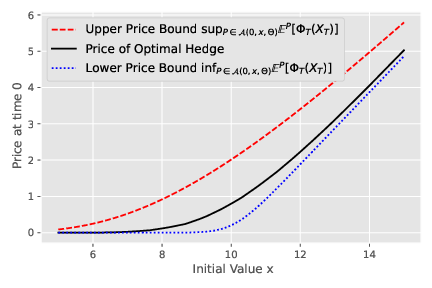}
\includegraphics[width=0.45\textwidth]{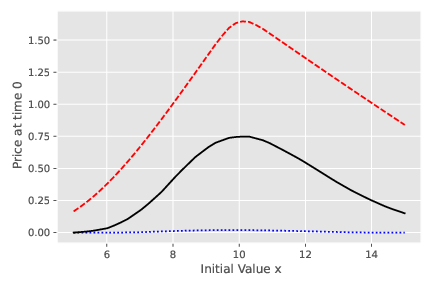}
\end{center}
\caption{For a call option with payoff function $\Phi_T(X_T)=(X_T-10)^+$ (left panel) and a butterfly option with payoff function $\Phi_T(X_T)=(X_T-8)^++(X_T-12)^+-2\cdot(X_T-10)^+$ (right panel), we compare the price bounds $\inf_{P\in \cA(0,x,\Theta)}\mathbb{E}^P[\Phi_T((X_T)]$ and $\sup_{P\in \cA(0,x,\Theta)}\mathbb{E}^P[\Phi_T((X_T)]$, that were computed by using a finite differences algorithm, with the price of the optimal hedge computed according to Algorithm~\ref{algo_nn}. We assume the parameters from \eqref{eq_parameters_nlap} and show prices as functions of the initial value $x$ of the stock price $X$.}\label{fig_comparison_prices}
\end{figure}

The results indicate that the price of the hedging strategies, that were computed according to Algorithm~\ref{algo_nn} and which lie well between lower and upper price bound, possess great practical relevance for two reasons. First, the associated prices are neither too low nor too high to be tradable. Second and in contrast to the prices computed as maximal expectations, the prices come with a trading strategy that allows to hedge the associated financial derivative under model uncertainty.

\section{Application to real-world data}\label{sec:data}

In this section, we evaluate the performance of the robust deep hedging strategy on financial data. To this end, we extracted daily closing prices of $20$ of the largest constituents\footnote{The considered constituents are: Apple Inc,  Microsoft Corporation,  Amazon.com Inc., Alphabet Inc. Class C, Berkshire Hathaway Inc. Class B, JPMorgan Chase $\&$ Co., Johnson $\&$ Johnson, Visa Inc. Class A, UnitedHealth Group Incorporated, NVIDIA Corporation, Procter $\&$ Gamble Company, Home Depot Inc., Mastercard Incorporated Class A, Bank of America Corp, Walt Disney Company, Comcast Corporation Class A, Exxon Mobil Corporation, Adobe Inc., Verizon Communications Inc., Intel Corporation.} of the US stock market index $S\&P~500$ from $26$ September $2008$ until $09$ April 2020 from \emph{Thomson Reuters Eikon}. This time period shows a high level of uncertainty during the beginning of the COVID-19 pandemic and thus poses a challenging environment for  hedging strategies. 

\subsection{Uncertainty in the parameter estimates} 
To analyze and illustrate the uncertainty present in parameter estimates, we consider parameter estimations on rolling windows. The obtained results allow us to specify the uncertainty set $\Theta$.
These results also underline the high degree of uncertainty present in the considered data. 

More precisely, we estimated the parameters under the assumption that the price observations follow generalized affine processes based on data from $26$ September 2008 until $03$ March $2020$
as follows: consider the discretization of the generalized affine processes $(X_t)_{t\geq 0}$ according to the Euler Maruyama scheme,  
$$
X_{i+1}=X_i+(b_0+b_1X_i)\Delta t_i +(a_0+a_1 X_i^+)^\gamma\Delta W_i
$$
for all values of the process $(X_i)_{1\leq i \leq N}$ on $N$ observation dates (daily observations), 
with time difference $\Delta t_i =1/250$, and normally distributed $\Delta W_i \sim N(0,\Delta t_i)$. 
Then, conditionally on $X_i$, $X_{i+1}$ is normally distributed since
$$
X_{i+1}\sim_{|X_i} \mathcal{N}\left(X_i+(b_0+b_1X_i)\Delta t_i, (a_0+a_1 X_i^+)^{2\gamma} \Delta t_i\right).
$$
Accordingly, given $N\in \N$ daily prices $x:=(x_1,\dots,x_N)$, the
log-likelihood function is given by
\begin{align*}
\ell_x(a_0,a_1,b_0,b_1,\gamma) &= \sum_{i=1}^{N-1}\log\left(\frac{1}{(a_0+a_1 x_i^+)^{\gamma}  \sqrt{2\pi
\cdot \Delta t_i}}\right) \\ &\hspace{1cm}-\frac{1}{2\Delta t_i} \left(\frac{x_{i+1}-x_i-(b_0+b_1x_i)\Delta t_i}{(a_0+a_1 x_i^+)^{\gamma}}\right)^2.
\end{align*}

We consider $2880$ trading days for each of the constituents of the $S\&P~500$-index. After every $100$ days we numerically maximize, by means of the \emph{Constrained Optimization by Linear Approximation}~(COBYLA) optimizer \cite[pp. 83-108.]{conn1997convergence}, the log-likelihood function $\ell_x$ w.r.t.\,the parameters $a_0$, $a_1$, $b_0$, $b_1$, $\gamma$, where $x=({x_1},\dots,{x_{250}})$ consists of the last $250$ trading days. The results of these estimations are illustrated for a single stock in the left panel of Figure~\ref{fig_parameter_estimation}. Moreover, in the middle panel of Figure~\ref{fig_parameter_estimation} we display all estimates from all of the considered $20$ constituents.

The obtained estimates show a considerable variation over time. For example, the estimator of $\gamma$ for Apple Inc.\ ranges from values slightly larger than $0.5$ to values around $1$ and highlights the advantage of using a \emph{generalized} affine process rather than a simple affine process where $\gamma$ would be fixed to $0.5$. The variations of all parameter estimates over the considered 20 constituents of the S\&P 500 confirm this finding. Also all other parameter estimates clearly exhibit a high degree of uncertainty.

\begin{figure}[t!]
\begin{center}
\includegraphics[width=0.32\textwidth]{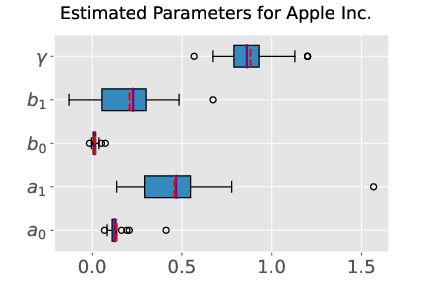}
\includegraphics[width=0.32\textwidth]{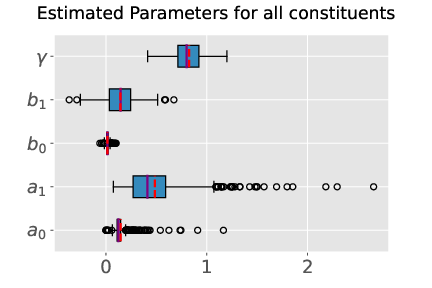}
\includegraphics[width=0.32\textwidth]{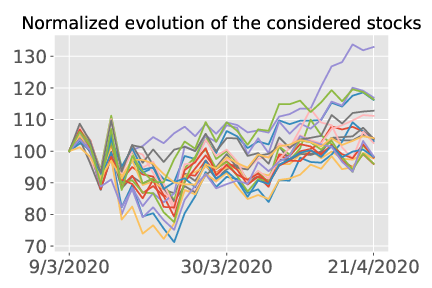}
\end{center}
\caption{\underline{Left:} parameters estimated by maximum-likelihood when assuming the stock of Apple Inc.\,follows a generalized affine process. The estimations are performed every $100$ days for a lookback window of $250$ days. \underline{Centre:} the maximum-likelihood-estimated parameters of all considered $20$ constituents of the $S\&P$ $500$. 
\underline{Right:} normalized (to initial value $100$) evolution of the considered $20$ constituents of the $S\&P$ $500$-index in the considered time period from $09$ March $2020$ until $21$ April $2020$.}
\label{fig_parameter_estimation}
\end{figure}

Given this time series of historical parameter estimates, we estimate the uncertainty set $\Theta$ by the obtained minima and maxima of the maximum-likelihood estimations. The obtained estimator is denoted by $\hat \Theta$. This represents a conservative approach and takes all past observations into account. More precisely, this constitutes the smallest possible choice given the past observations. Of course, the uncertainty set could also be increased to improve robustness, which however comes at the cost of higher (and therefore potentially less attractive) derivatives' prices and higher hedging costs.

\begin{remark}[Historical measure vs risk-neutral measure]
In this section we are mainly interested in the hedging performance which is typically evaluated under the %
\emph{historical} measure. However, there are also cases where one prefers the distribution under the \emph{risk-neutral measure}, see for example \cite{FoellmerSchied2004} for a detailed exposition on various hedging concepts. In the later case one would obtain parameter estimates from liquid derivatives' prices through calibration and then proceed analogously. 
\end{remark}

\begin{remark}[Choice of $\hat \Theta$]
We defined the parameter set $\hat \Theta$ by the intervals induced by the extreme maximum-likelihood-estimates. This corresponds to a conservative approach in which even outliers are deemed to be relevant for the future evolution of the underlying stochastic process. In less conservative approaches one could instead take into account inter-quartile ranges of the estimated parameters or only a grid of parameters associated to the historical estimates. The latter approach avoids that parameter combinations which did not appear in the past (e.g. large values of $\gamma$ and $a_1$ usually do not occur at the same time) are considered as relevant for the future evolution.
\end{remark}

Relying on the estimated uncertainty set $\hat \Theta$ we compute, according to Algorithm~\ref{algo_nn} a hedging strategy for an Asian at-the money put option with daily observations:
\begin{equation}\label{eq_payoff_asian}
\Phi_T=\left(x_0-\frac{1}{30} \sum_{t=1}^{30} X_t\right)^+,
\end{equation}
where $x_0$ corresponds to the respective initial spot value at $09$ March $2020$  and $T=30$ trading days (i.e.\,maturity $21$ April $2020$). Further, we compute for each constituent a hedging strategy which only takes the last maximum-likelihood estimation of the last $250$ days into account.
We then evaluate, on the real price evolution of the constituents of the $S\&P~500$ from $09$ March $2020$ until $21$ April $2020$ (compare the right panel of Figure~\ref{fig_parameter_estimation}) how both hedging strategies perform.

\begin{table}[t!]
\begin{center}\resizebox{\linewidth}{!}{
\begin{tabular}{ll*{6}{c}{l}} \toprule
         Parameters     & fixed & robust & $a_0$ fixed & $a_1$ fixed & $b_0$ fixed & $b_1$ fixed & $\gamma$ fixed &Black--Scholes\\
\midrule
mean &5.8939  &\textbf{0.9014}  &0.9369   &5.1408   &0.9030   &0.9433  &2.1390 &5.7859\\
std. dev. &2.7666   &0.7705   &0.7410  &3.0596  &0.7545   &0.8073 &1.8934 &2.8477\\
min. &0.4391 &0.0279  &0.0104  &0.4116  &0.0192   &0.0185  &0.2303 &2.0983\\
max.& 12.5595  &2.7106   &2.6356  &12.3456  &2.4122  &2.7329   &7.4269 &12.0025\\ \bottomrule
\end{tabular}}
\end{center}
\caption{Relative hedging errors  for an Asian at-the money put option $\left(x_0-\frac{1}{30} \sum_{t=1}^{30} X_t \right)^+$ of trained hedging strategies of the considered $20$ constituents. Each column represents another trained strategy which considers either fixed parameters (to the most recent maximum-likelihood estimation), robust parameter intervals (determined by the most extreme maximum-likelihood estimations), or robust intervals except for a single parameter which is still fixed. The rightmost column shows the hedging error when assuming an underlying Black--Scholes model.}\label{tbl_parameters_sp500}
\end{table}

For this we compare the relative hedging error of the strategies. 
The results are depicted in Table~\ref{tbl_parameters_sp500} and indicate that the robust hedging strategy may, in particular, perform better in periods with high volatility as in the period under consideration. In Table~\ref{tbl_parameters_sp500} we further display the hedging error of strategies that were trained, when all of the considered parameters are assumed to be contained in intervals except for a single parameter which is fixed. 
This analysis allows to compare and analyse the effect of robustness of single parameters on the hedging error. We observe that taking uncertainty into account is in particular important for the volatility parameters and the parameter $\gamma$. More precisely,  while fixing the drift parameters does not lead to considerably worse hedging errors, fixing the exponent $\gamma$ of the volatility term, and in particular the volatility parameter $a_1$, significantly increases the mean and the standard deviation of the hedging error.

For comparison we also report the hedging error when applying the deep hedging approach in a (non-robust) Black--Scholes model\footnote{Note that the Black--Scholes model can be considered as a special case of an NGA-process through setting $\Theta=\{0\}\times\{\mu\}\times \{0\} \times \{\sigma\}\times \{1\}$ for some mean $\mu\in \R$ and some variance $\sigma^2>0$.}, where the parameters are estimated in a consistent manner through maximum likelihood estimation %
while taking into account the time series of the last $250$ trading days. The results of this hedging approach are depicted in the rightmost column of Table~\ref{tbl_parameters_sp500} and show that hedging under the Black--Scholes model leads in the considered period to a mean hedging error and standard deviation comparable to the mean hedging error of an NGA-process with fixed parameters.

This supports our choice of considering the class of generalized affine models in the robust pricing and hedging approach especially during such periods of market turmoil.

\subsection{Considering a non-crisis setting}
While we have now provided evidence for the outperformance of the robust hedging approach over other approaches in a crisis period, the question arises whether the approach is flexible enough to perform comparable to other approaches in periods that would be rather classified as \emph{non-crisis} periods. 

To this end, and to be consistent with the previously introduced methodology, we consider three additional $30$ day testing periods, starting $100$ trading days, $200$ trading days and $300$ trading days, respectively, after $09$ March $2020$. For each additional period we take new maximum-likelihood-estimations of the last $250$ days into account and evaluate, for the same payoff function \eqref{eq_payoff_asian}, the performances of a robust hedging approach, of a non-robust hedging approach and of a hedging approach under a Black--Scholes model. The results of this study are displayed in Table~\ref{tbl_parameters_sp500_after_crisis} and show particularly that in these periods the mean hedging errors and the standard deviation of all approaches are reduced in comparison with the crisis period.

\begin{table}[t!]
\begin{center}{
\begin{tabular}{ll*{1}{c}{c}} \toprule
         Parameters     & fixed & robust &Black--Scholes\\
\midrule
mean  &\textbf{0.4424}	&0.5565 &0.4627\\
std. dev.  &0.2272	&0.4380 &0.2291\\
min.  &0.0265 	&0.0781 &0.0170\\
max. &0.9608 	&2.9152 &0.8612\\ \bottomrule
\end{tabular}}
\end{center}
\caption{\emph{(Non-crisis period):} Relative hedging errors for an Asian at-the money put option $\left(x_0-\frac{1}{30} \sum_{t=1}^{30} X_t \right)^+$ of trained hedging strategies of the considered $20$ constituents. We take three testing periods into account, starting $100$ trading days, $200$ trading days and $300$ trading days, respectively, after $09$ March $2020$ which was the initial day for the period considered in Table~\ref{tbl_parameters_sp500}.}\label{tbl_parameters_sp500_after_crisis}
\end{table}

The best performing model in these periods turns out to be the NGA-process with fixed parameters, while pursuing a robust hedging approach leads to a slightly higher hedging error and a higher standard deviation. These observations indicate that non-robust approaches perform best in regular out-of-crisis periods, whereas a robust hedging approach performs slightly worse in these periods, presumably since such hedging strategies are adjusted and calibrated to a broader range of possible future market movements. Our investigation of the performance of the hedges in the crisis period (Table~\ref{tbl_parameters_sp500}) however reveals that this broad calibration can provide additional strong protection against unexpected market movements as they can be observed in crises.

\subsubsection{Relation to Bayesian approaches}\label{rem_bayesian}
Instead of the presented frequentist approach, in which we use the minimal and maximal maximum-likelihood-estimations to determine the intervals representing $\Theta$ and then to assume that the parameters of the SDE are uniformly distributed on $\Theta$, one might also consider other distributions: first, the empricial distribution of the parameter estimates as shown in Figure \ref{fig_parameter_estimation} is a natural choice. Second,  a Bayesian approach for the determination of the parameter intervals can be implemented, see \cite{duembgen2014estimate} for a Bayesian approach. In this  approach one starts from a prior distribution (e.g. uniform on some pre-defined intervals) for all of the parameters and then sequentially updates the resulting posterior distributions contingent on the same data which we use for the maximum-likelihood-estimations. Eventually, to determine optimal hedging strategies one modifies Algorithm~\ref{algo_nn} by drawing parameters according to the obtained posterior distributions, as already detailed in Remark \ref{rem:prior distributions}. Alternatively, quasi Bayesian approaches as in \cite{BrignoneGonzatoLuetkebohmert} can be used where an asymptotic distribution of parameters is estimated from the quasi posterior distribution. Since the Bayesian approach may put relatively few weight to extreme parameters we decided to implement the presented approach which puts equal weight to all of the parameters that are considered possible. This approach is therefore \emph{robust} w.r.t.\,extreme market movements, for what we provide evidence in the example in Section~\ref{sec:data}.

\section{Conclusion}\label{sec:conclusions}

In this work we studied parameter uncertainty in the class of generalized affine processes and developed a robust hedging approach relying on deep neural networks. This approach shows resilience against unexpected changes in the dynamics of the underlying, which justifies the claimed robustness of this method. Our research is a first step towards the practical application of robust hedging approaches and still many questions remain open: the most pressing one is the practical determination of the uncertainty interval - how much risk is one willing to take by considering a smaller interval (which clearly in good weather conditions will be cheaper in pricing and hedging)? 
The second, highly interesting question is to incorporate transaction costs into the robust deep hedging approach, to treat other dynamics of the underlying and to consider  loss functions different to the quadratic one we used in this paper.

\appendix
\section*{Acknowledgements}
We thank the editor and two anonymous referees for several comments which significantly improved our paper. Moreover, we are thankful to David Criens for helpful remarks. 
Financial support of the NAP Grant \emph{Machine Learning based Algorithms in Finance and Insurance} and from Deutsche Forschungsgemeinschaft (DFG) of the grant \emph{SCHM 2160/13-1} is gratefully acknowledged. 

\section{Appendix}

\subsection{Existence of generalized affine diffusions}

It is well known that the state space $E$ needs to be chosen in correspondence with $\Theta$. In the case where $E=\R$, this does not pose difficulties, but in the case where $E=\R_{>0}$ some care has to be taken. The special case where $\gamma = \nicefrac 1 2$ is the content of Proposition 1 in \cite{FadinaNeufelSchmidt2019}.

We call the state space $E$ \emph{proper} for the non-linear generalized affine process $\cA(t,x,\Theta)$ if $P(X_s \in E, t \le s \le T) = 1$ for all $P \in \cA(t,x,\Theta)$ and all $0 \le t \le T,\ x \in E$. 
The next lemma extends Proposition 1 in \cite{FadinaNeufelSchmidt2019} to the case where $\gamma \neq  \nicefrac 1 2$.

 \begin{lemma}\label{lem:proper}
 Assume that $E=\R_{>0}$, $\underline b_0>0$, $\underline a_0= \bar a_0 = 0$, $\underline a_1>0$ and $\nicefrac 1 2 < \underline \gamma \le \bar \gamma \le 1$. Consider the NGA  $\cA(0,x_0,\Theta)$ with $x_0 \in E$. Then it holds for any $P\in \cA(0,x_0,\Theta)$ that 
 $$ P(X_s >0, \ t \le s \le T ) = 1. $$
 \end{lemma}
 \begin{proof}
  For the proof we rely on the integral test proposed in Theorem 5.2 in \cite{criens2020no}. To this end we consider a sufficiently small subset $(0,\varepsilon) \subset E$ such that $b_0+b_1 x > 0$ for all $x \in (0,\varepsilon)$. 

  To begin with, we  observe the estimates
  \begin{align}
  \begin{aligned}
    (a_1 x)^{2 \gamma} & \le  (a_1 \varepsilon)^{2 \gamma} =: \bar a, \\
    \frac{b_0+b_1 x}{(a_1 x)^{2 \gamma}} & \ge \frac{\inf_{y \in (0,\varepsilon)}(b_0+b_1y)}{a_1^{2 \gamma}} x^{-2 \gamma} =: u_0 x^{-2 \gamma} =:\underline u(x)
  \end{aligned}
  \end{align}
  with constants $\bar a >0$ and $u_0 >0$.

  The next step is to show that  $v(\underline u, \bar a)(x)$ from Equation (5.5) in \cite{criens2020no}  explodes as $x \to 0$, where
  \begin{align}\label{temp 873}
    v(\underline u, \bar a)(x) = \int_{x_0}^x \exp\Big( -2 \int_{x_0}^y \underline u(z) dz\Big) \int_{x_0}^y \frac{2 \exp\Big(\int_{x_0}^u 2 \underline u(z) dz\Big)}{\bar a} du \, dy.
  \end{align}
   So in the following we consider $x<\nicefrac{x_0}2$. Then,  
  \begin{align}
    \eqref{temp 873} &\ge \frac{2}{\bar a} \int_{\nicefrac{x_0}{2}}^x \exp\Big( -2 \int_{x_0}^y \underline u(z) dz\Big) \int_{x_0}^y \exp\Big(\int_{x_0}^u 2 \underline u(z) dz\Big) du \, dy.
  \end{align}
  Then we can estimate (since $y < \nicefrac{x_0}2$), setting $\beta = 2 \gamma -1 >0$,
  \begin{align}
    \int_{x_0}^y \exp\Big(\int_{x_0}^u 2 \underline u(z) dz\Big) du &\ge \int_{x_0}^{\nicefrac{x_0}2} \exp\Big(\int_{x_0}^u 2 \underline u(z) dz\Big) du\notag\\
    & =  \int_{x_0}^{\nicefrac{x_0}2}  \exp\Big( 2 u_0 \frac{(x_0)^{-\beta} - u^{-\beta}}{\beta}  \Big) du \notag \\
    & \ge \frac{x_0}{2}  \exp\Big( 2 u_0 \frac{(x_0)^{-\beta} - (\nicefrac{x_0}2)^{-\beta}}{\beta} \Big) =: A_1 \label{eqA2}
    \end{align}
    for some constant $A_1>0$.  Up to constants we can now estimate $v(\underline u, \bar a)$ from below by
  \begin{align*}
     \int_{\nicefrac{x_0}{2}}^x \exp\Big( -2 \int_{x_0}^y u_0 z^{-2 \gamma} dz\Big)   dy 
     &= e^{\frac{-2u_0}{\beta x_0^{\beta}}}  \int_{\nicefrac{x_0}{2}}^x \exp\Big( \frac{2 u_0}{\beta} y^{-\beta} \Big) dy \\
     &= \frac{-1}{\beta} e^{\frac{-2u_0}{\beta x_0^{\beta}}} \int_{(\nicefrac{x_0}{2})^{-\beta}}^{x^{-\beta}} e^{\frac{2 u_0}{\beta}  z }  z^{\beta'} dz
  \end{align*}
  with $\beta'=-\beta^{-1}-1.$ Now it is easy to see that the integral on the right hand side explodes as $x\to 0$ by l'Hospital's rule.
 \end{proof}

As a consequence of Lemma \ref{lem:proper} we obtain that the state space $E$ is proper in the following cases:
\begin{enumerate}[(i)]
	\item $E=\R$ and $\underline a_0 >0$, %
	\item $E=\R_{>0}$, $\gamma = \nicefrac 1 2$ : $\underline b_0 > 0$, $\underline a_0>0$,  and $\underline b_0 >\bar a_1/2$, 
	\item $E=\R_{>0}$, $\nicefrac 1 2 < \underline \gamma \le \bar \gamma \le 1$: $\underline b_0>0$, $\underline a_0= \bar a_0 = 0$, and $\underline a_1>0$.
\end{enumerate}
The next proposition establishes conditions such that the set of semimartingale measures $\cA(t,x,\Theta)$ is not empty.
\begin{proposition}
	[Existence of generalized affine process]\label{prop:existence}
Let $\gamma\in [1/2,1]$. If $E=\mathbb{R}$ assume $b_0,\ a_0>0$ and $a_1=0$ while for $E=\mathbb{R}_{>0}$ we assume $b_0 >0, a_0=0$ and $a_1>0$ and, for $\gamma = 1/2$, additionally  $b_0 > a_1/2$. Then for all $t\in[0,T]$ and $x\in E$ there exists a unique strong solution to the SDE (\ref{equ generalized affine process}).
\end{proposition}

\begin{proof}
The theorem follows from Corollary 5.5.16 in \cite{KaratzasShreve1991} using the results from Engelbert and Schmidt (\cite{engelbert1985one,engelbert1985solutions}). First note that in the case $E=\mathbb{R}_{>0}$ the function $1/(a_0+a_1 x)^{2\gamma}$ is locally integrable for any $x\in \mathbb{R}_{>0}$ if $a_0=0$ and $a_1>0$. If $E=\mathbb{R}$, the local integrability follows because $a_0>0$ and $a_1=0$ in that case. Further, for any $x,y\in \mathbb{R}$ we have
$$
|(b_0+b_1x)-(b_0+b_1y)|\leq \kappa |x-y|
$$
where
$\kappa\equiv \max\{|a_0|,|a_1|,|b_0|,|b_1|\}\in \mathbb{R}.$
Moreover, we have
$$
|(a_0+a_1x)^{\gamma}-(a_0+a_1 y)^{\gamma}|\leq \kappa |x-y|^{\gamma},
$$
i.e. the function $h$ in Corollary 5.16 in \cite{KaratzasShreve1991} is given by the strictly increasing function $h(z)=\kappa z^\gamma$ with $h(0)=0$. Since we chose $\gamma\in [1/2,1]$, the function $h$ satisfies the condition
$$
\int_{(0,\epsilon)} h^{-2}(u)du=\infty\quad \forall \epsilon>0.
$$
Further, the conditions 
\begin{itemize}
\item[(ND)] $(a_0+a_1 x)^{2\gamma}>0$ for all $x\in E$ and
\item[(LI)] for all $x\in E$ there exists an $\epsilon>0$ such that $\int_{x-\epsilon}^{x+\epsilon} \frac{|b_0+b_1 y|}{(a_0+a_1 y)^{2\gamma}} dy <\infty$
\end{itemize} 
in \cite{KaratzasShreve1991} are satisfied when we choose $a_0>0,\, a_1=0$ if $E=\mathbb{R}$ and $a_0=0, \,a_1>0$ if $E=\mathbb{R}_{>0}$. Thus, there exists a strong solution to the SDE (\ref{equ generalized affine process}), possibly up to an explosion time. Explosions to $+\infty$ in finite time do not occur since we have at most linear growth. 

If the state space is $\R_{>0}$ and $\gamma = \frac 1 2$, then the process $X$ does not reach zero due to Proposition 1 in \cite{FadinaNeufelSchmidt2019}. If $\gamma \in (1/2,1]$, Lemma \ref{lem:proper} implies that again $X$ does not reach zero and the conclusion follows.
\end{proof}

\subsection{Proof of the nonlinear Kolmogorov equation}

In this section, we prove Theorem \ref{thm:Kolmogorov}, which we repeat for the reader's convenience.

\begin{customthm}{2.2}
Consider a family of nonlinear generalized affine processes with state space $E$ and a Lipschitz continuous payoff function $\psi:E\rightarrow\mathbb{R}$. Then,
$$
v(t,x):=\sup_{P\in \cA(t,x,\Theta)}\mathbb{E}^P[\psi(X_T)],\quad x\in E
$$
is a viscosity solution to the PDE (\ref{equ nonlinear PDE}).
\end{customthm}

For the proof we will need some preliminary tools.

\begin{lemma}\label{lemma 1} 
Let $\gamma \le 1$. 
For all $q\geq 1$ there exists an $0<\epsilon\equiv \epsilon(q)<1$ such that for all $0<h\leq \epsilon$, all $t\in[0,T-h]$ and $x\in E$ it holds that
$$
\sup_{P\in\cA(t,x,\Theta)} \mathbb{E}^P\left[\sup_{0\leq s\leq h} |X(t+s)-x|^q\right]\leq c \big(h^{q/2}+h^q\big)
$$
for some constant $c=c(x,q)>0$.
\end{lemma}

The proof is a modification of the proof of Lemma 3 in \cite{FadinaNeufelSchmidt2019} and Lemma 5.2 in \cite{NeufeldNutz2017} and takes the generalized setting into account. %

\begin{proof}
	   Consider $P \in \cA(t,x,\Theta)$ and denote by $X_s=x+B^P_s+M^P_s$, $s \ge t$, the semimartingale representation of $X$ from Equation \eqref{eq:sem mart char}. We will repeatedly use the elementary inequality  
   \begin{align}\label{eq:elementaryineq} 
        (a_1+a_2)^q \le 2^{q-1}( a_1^q + a_2^q) 
   \end{align}
   and denote $c_q:=2^{q-1}$.

   First, the  Burkholder--Davis--Gundy (BDG) inequality (see Theorem IV.4.1 in  \cite{RevuzYor}) together with Jensen's inequality and \eqref{eq:elementaryineq} yields for any $h \in [0,T-t]$ that
	\begin{align}\label{temp:308}
	   \E^P\Big[\sup_{0\leq s \leq h}|X_{t+s}-x|^q\Big] &\leq c_q \E^P\Big[\sup_{0\leq s \leq h}|M^P_{t+s}|^q\Big] + c_q \E^P\Big[\sup_{0\leq s \leq h}|B^P_{t+s}|^q\Big]\\
	       &\leq c_q \widetilde C_q \E^P\bigg[\Big(\int_t^{t+h} \alpha_u\,du\Big)^{\nicefrac q 2}\bigg] + c_q\E^P\bigg[\Big(\int_t^{t+h} |\beta^P_u|\,du\Big)^q\bigg].\nonumber
	\end{align}
    Note that the constant $\widetilde C_q \ge 1$ from the BDG inequality does depend on $q$ only.

	We define $ \mathcal{K}= 1+ |\underline{b}^0|+|\underline{b}^1|+|\bar{b}^0|+|\bar{b}^1|+\bar{a}^0+\bar{a}^1$ and choose any $0<\varepsilon=\varepsilon(q)<1$
	small enough such that it satisfies 
	\begin{align} \label{cond:epsilon}
	1-c_q^3\widetilde C_q {\mathcal{K}}^q  (\varepsilon^q+\varepsilon^{q/2})>0. \end{align} 
	Let us verify that such a fixed $\varepsilon$ satisfies the desired property: 
	by the very definition of $P \in \cA(t,x,\Theta)$, we have on $[t,t+h]$ that both $\alpha$ and $|\beta^P|$ are bounded from above by 
	$({\mathcal{K}}+{\mathcal{K}}\sup_{0\leq s \leq h}|X_{t+s}|)^{2 \gamma}\ge 1$ and ${\mathcal{K}}+{\mathcal{K}}\sup_{0\leq s \leq h}|X_{t+s}|\ge 1$, respectively, since they are GA-dominated. 
	This, together with  Jensen's inequality, yields that
	\begin{align}
	 \E^P\bigg[\Big(\int_t^{t+h} \alpha_u\,du\Big)^{\nicefrac q 2}\bigg] 
	    &\leq  h^{q/2}\E^P\bigg[\Big({\mathcal{K}} + {\mathcal{K}} \sup_{0\leq s \leq h}|X_{t+s}|\Big)^{2 \gamma \nicefrac q 2}\bigg] \label{temp:542} \\
	    &\leq  h^{q/2}\E^P\bigg[\Big({\mathcal{K}} + {\mathcal{K}} \sup_{0\leq s \leq h}|X_{t+s}|\Big)^{\gamma q}\bigg]. \notag
	\end{align}
	Since $ {\mathcal{K}} + {\mathcal{K}} \sup_{0\leq s \leq h}|X_{t+s}| \ge 1$ and $\gamma \le 1$, we have that 
	$$ \Big({\mathcal{K}} + {\mathcal{K}} \sup_{0\leq s \leq h}|X_{t+s}|\Big)^{\gamma q} \le \Big({\mathcal{K}} + {\mathcal{K}} \sup_{0\leq s \leq h}|X_{t+s}|\Big)^{q}. $$
	 Then,
	\begin{align*}\eqref{temp:542}
	&\leq  h^{q/2}c_q\bigg({\mathcal{K}}^{q }  + {\mathcal{K}}^{q} \E^P\Big[\Big(\sup_{0\leq s \leq h}|X_{t+s}|\Big)^q\Big]\bigg)\\
	    &\leq  h^{q/2}c_q^2\bigg({\mathcal{K}}^q  + {\mathcal{K}}^q |x|^q + {\mathcal{K}}^q \E^P\Big[\Big(\sup_{0\leq s \leq h}|X_{t+s}-x|\Big)^q\Big]\bigg)
	\end{align*}
 Since the drift is affine dominated, we obtain in a similar way that 
	\begin{align*}
	 \E^P\bigg[\Big(\int_t^{t+h} |\beta^P_u|\,du\Big)^q \bigg] 
	    &\leq  h^q\E^P\Big[\Big({\mathcal{K}}+{\mathcal{K}}\sup_{0\leq s \leq h}|X_{t+s}|\Big)^q\Big]\\
	    &\leq h^{q} c_q^2\bigg({\mathcal{K}}^q  + {\mathcal{K}}^q |x|^q + {\mathcal{K}}^q \E^P\Big[\Big(\sup_{0\leq s \leq h}|X_{t+s}-x|\Big)^q\Big]\bigg).
	\end{align*}
	Inserting these inequalities into \eqref{temp:308}, considering $h\leq \varepsilon$, and noting that  $\tilde C_q \ge 1$ 
	implies that
	\begin{align}\label{eq:391}
	\lefteqn{\E^P\Big[\sup_{0\leq s \leq h}|X_{t+s}-x|^q\Big]}\hspace{1.5cm}\\
	\leq & 	c_q^3 \, \widetilde C_q {\mathcal{K}}^q(h^{q/2} + h^q) \E^P\Big[\sup_{0\leq s \leq h}|X_{t+s}-x|^q\Big]
	       +c_q^3 \, \widetilde C_q  {\mathcal{K}}^q (1+|x|^q)(h^{q/2} + h^q) \nonumber\\
	\leq & 	c_q^3 \, \widetilde C_q {\mathcal{K}}^q(\varepsilon^{q/2} + \varepsilon^q) \E^P\Big[\sup_{0\leq s \leq h}|X_{t+s}-x|^q\Big]
	       +c_q^3 \, \widetilde C_q  {\mathcal{K}}^q (1+|x|^q)(h^{q/2} + h^q).\nonumber
	\end{align}
	Since $h\leq \varepsilon$ and we chose $0<\varepsilon<1$ such that \eqref{cond:epsilon} holds, we obtain for the constant 
	$c:=\frac{c_q^3 \, \widetilde C_q  {\mathcal{K}}^q (1+|x|^q)}{1-c_q^3 \, \widetilde C_q {\mathcal{K}}^q(\varepsilon^{q/2} + \varepsilon^q)}>0$, being  independent of $t,h,P$, that
	\begin{equation*}
	     \E^P\Big[\sup_{0\leq s \leq h}|X_{t+s}-x|^q\Big] \leq c\,\big(h^{q/2}+h^q\big).
	\end{equation*}
	As $P \in \cA(t,x,\Theta)$ was  chosen  arbitrarily, the claim is proven.
\end{proof}

\begin{lemma}\label{le:value-funct-cont}
Consider a nonlinear generalized affine process $\cA(0,x,\Theta)$ and a derivative with Lipschitz-continuous payoff function $\psi: E\rightarrow \mathbb{R}$. Then the value function 
$$
v: [0,T]\times E\rightarrow \mathbb{R}, \, (t,x)\mapsto v(t,x)
$$
is jointly continuous. In particular, $v(t,x)$ is locally $\nicefrac 12$-H\"{o}lder continuous in $t$ and Lipschitz-continuous in $x$.
\end{lemma}

\begin{proof}
The statement follows similarly to Lemma 4 in \cite{FadinaNeufelSchmidt2019} and Lemma 5.3 in \cite{NeufeldNutz2017}.
For $x\neq y$ and fixed $t\in[0,T]$ it holds that
$$
|v(t,x)-v(t,y)|\leq \sup_{P\in \cA(t,x,\Theta)}\mathbb{E}^P
\left[ |\psi(X_T)-\psi(y-x+X_T)|\right]\leq L|y-x|
$$
where $L$ is the Lipschitz constant of the function $\psi$. Thus, the value function is Lipschitz-continuous in $x$.

For the locally $\gamma$-H\"{o}lder continuity, let $t\in[0,T)$ and $0\leq u\leq T-t$ small enough. Then the Lipschitz-continuity, the dynamic programming principle in Proposition \ref{prop DPP} and Lemma \ref{lemma 1} imply that
$$
\begin{array}{ccl}
|v(t,x)-v(t+u,x)| &\leq& \left|\sup_{P\in\cA(t,x,\Theta)} \mathbb{E}^P\left[v(t+u,X_{t+u})-v(t+u,x)\right]\right|\\[2mm]
&\leq & L\cdot \sup_{P\in\cA(t,x,\Theta)} \mathbb{E}^P\left[ |X_{t+u}-x| \right]\\[2mm]
&\leq& L\cdot c\cdot (u+u^{\nicefrac 12})
\end{array}
$$
with the constant $c=c(x,1)$ from Lemma \ref{lemma 1}.
Choosing a sequence $(t_n,x_n)$ converging to $(t,x)$ we have that
$$
\begin{array}{ccl}
|v(t_n,x_n)-v(t,x)| &\leq& |v(t_n,x_n)-v(t_n,x)|+|v(t_n,x)-v(t,x)|\\
&\leq& L |x_n-x|+L c(x,1) \left( |t_n-t|^{\nicefrac 12} +|t_n-t|\right)
\end{array}
$$
The statement follows for $n\rightarrow \infty.$
\end{proof}

\begin{proof} \emph{(of Theorem \ref{thm:Kolmogorov})}
	The proof essentially follows the well-known standard arguments in stochastic control, see e.g., the proof of \cite[Proposition~5.4]{NeufeldNutz2017}.
	
	By Lemma \ref{le:value-funct-cont}, $v(t,x)$ is continuous on $[0,T)\times \R$, and we have $v(T,x)= \psi(x)$ by the definition of $v$. We show that $v$ is a viscosity subsolution of the nonlinear affine PDE defined in \eqref{eq:def:PDE}; the supersolution property is proved similarly. We remark that in the subsequent lines within this proof, $c>0$ is a constant whose values may change from line to line.
	
	Let $(t,x) \in [0,T)\times \R$ and let $\varphi \in C^{2,3}_b([0,T)\times \R^d)$ be such that $\varphi\geq v$ and $\varphi(t,x)=v(t,x)$. By the dynamic programming principle obtained in Proposition~\ref{prop DPP}, we have for any $0<u<T-t$ that
	\begin{align}\label{eq:PDE-DPP}
	    0&=\sup_{P\in \cA(t,x,\Theta)}\E^P\big[v(t+u,X_{t+u})-v(t,x)\big] \nonumber\\
	     &\leq \sup_{P\in \cA(t,x,\Theta)}\E^P\big[\varphi(t+u,X_{t+u})-\varphi(t,x)\big]. 
	\end{align}
	Fix any $P\in \cA(t,x,\Theta)$, denote as above by  $(\beta^P,\alpha)$ the differential characteristics of the continuous semimartingale $X$ under $P$,
  and denote by $M^P$ the $P$-local martingale part of the $P$-semimartingale $X$.
	Then, It\^o's formula yields
	\begin{align}
	\varphi (t+u,&X_{t+u}) - \varphi(t,x) =
	\int_0^u \partial_t \varphi(t+s,X_{t+s}) \,ds
	+
	\int _0^u \partial_x \varphi(t+s,X_{t+s}) \,dM^P_{t+s} \nonumber\\
	& + \int_0^u \partial_x \varphi(t+s,X_{t+s}) \beta_{t+s}^P \,ds 
	+ \frac{1}{2} \int_0^u \partial_{xx} \varphi (t+s,X_{t+s}) \alpha_{t+s} \,ds. \label{eq:PDE1}
	\end{align}
	As $\varphi \in C_b^{2,3}([0,T)\times \R)$,  $\partial_x \varphi$ is uniformly bounded,%
	we see that for small enough $0<u<T-t$ the local martingale part in \eqref{eq:PDE1} is in fact a true martingale, starting at $0$. In particular, its expectation vanishes. 	The next step is to estimate the expectation of the other terms. In this regard, note that
	\begin{align}
	 \E^P\bigg[ \int_0^u &\partial_x \varphi(t+s,X_{t+s}) \beta^P_{t+s}\, ds\bigg]  \nonumber\\
	\leq & \ \int_0^u \E^P \bigg[ \big| \partial_x \varphi(t+s,X_{t+s}) - \partial_x \varphi (t,x)\big| \, |\beta^P_{t+s}| +
	\partial_x \varphi(t,x) \beta^P_{t+s} \bigg]\, ds. \label{eq:PDE2}
	\end{align}
	Since $\varphi \in C^{2,3}_b$, $\partial_x \varphi$ is Lipschitz. Hence, we obtain with the constant $\mathcal K = 1+ |\underline{b}^0|+|\underline{b}^1|+|\bar{b}^0|+|\bar{b}^1|+\bar{a}^0+\bar{a}^1$  together with 
	 Lemma~\ref{lemma 1} that for small enough $u$,
	\begin{align}
	 \ \int_0^u \E^P \Big[ \big| &\partial_x \varphi(t+s,X_{t+s}) - \partial_x \varphi (t,x)\big| \cdot |\beta_{t+s}^P| \Big]\, ds \nonumber\\
	\leq & \  c\int_0^u \E^P \Big[ \big(s+\sup_{0\leq v\leq u}|X_{t+v}-x|\big)  \cdot |\beta^P_{t+s}| \Big]\, ds \nonumber\\
	\leq & \  c\int_0^u \E^P \Big[ \big(s+\sup_{0\leq v\leq u}|X_{t+v}-x|\big)\, \big(\mathcal{K} +\mathcal{K} \sup_{0\leq v\leq u}|X_{t+v}|\big)\Big]\, ds \nonumber\\
	\leq & \  c\int_0^u \E^P \Big[ \big(s+\sup_{0\leq v\leq u}|X_{t+v}-x|\big)\, \big(\mathcal{K} +\mathcal{K} |x| +\mathcal{K} \sup_{0\leq v\leq u}|X_{t+v}-x|\big)\Big]\, ds \nonumber\\
	\leq & \ c\big(u^3 +u^{5/2}+ u^2 + u^{3/2}\big). \label{eq:PDE3}
	\end{align}
	Inserting \eqref{eq:PDE3} into \eqref{eq:PDE2} yields
	\begin{align} \label{eq:PDE4}
	  	   \lefteqn{ \E^P \bigg[ \int_0^u \partial_x \varphi(t+s,X_{t+s}) \beta^P_{t+s}\, ds\bigg]} \hspace{2cm} \nonumber\\
	   & \leq \int_0^u \E^P\Big[\partial_x \varphi(t,x) \,\beta^P_{t+s}\Big]\, ds + c\big(u^3 +u^{5/2}+ u^2 + u^{3/2}\big). 
	\end{align}
	The same argument applied to $\partial_{xx} \varphi$ leads to
	\begin{align}
	 \ \int_0^u \E^P \Big[ \big| &\partial_{xx} \varphi(t+s,X_{t+s}) - \partial_{xx}  \varphi (t,x)\big| \cdot |\alpha_{t+s}| \Big]\, ds \nonumber\\
	\leq & \  c\int_0^u \E^P \Big[ \big(s+\sup_{0\leq v\leq u}|X_{t+v}-x|\big)  \cdot |\alpha_{t+s}| \Big]\, ds \nonumber\\
	\leq & \  c\int_0^u \E^P \Big[ \big(s+\sup_{0\leq v\leq u}|X_{t+v}-x|\big)\, \big(\mathcal{K} +\mathcal{K} \sup_{0\leq v\leq u}|X_{t+v}|\big)^\gamma\Big]\, ds \nonumber\\
	\leq & \  c\int_0^u \E^P \Big[ \big(s+\sup_{0\leq v\leq u}|X_{t+v}-x|\big)\, \big(\mathcal{K} +\mathcal{K} \sup_{0\leq v\leq u}|X_{t+v}|\big)\Big]\, ds \nonumber\\
	\leq & \  c\int_0^u \E^P \Big[ \big(s+\sup_{0\leq v\leq u}|X_{t+v}-x|\big)\, \big(\mathcal{K} +\mathcal{K} |x| +\mathcal{K} \sup_{0\leq v\leq u}|X_{t+v}-x|\big)\Big]\, ds \nonumber\\
	\leq & \ c\big(u^3 +u^{5/2}+ u^2 + u^{3/2}\big). %
	\end{align}
	and we obtain that
	\begin{align} \label{eq:PDE5}
	\lefteqn{\E^P\bigg[ \int_0^u \partial_{xx} \varphi(t+s,X_{t+s})\, \alpha_{t+s}\, ds\bigg] } \hspace{2cm} \nonumber\\
	    & \leq \int_0^u \E^P\Big[\partial_{xx} \varphi(t,x)\, \alpha_{t+s}\Big]\, ds + c\big(u^3 +u^{5/2}+ u^2 + u^{3/2}\big).
	\end{align}
	Moreover, by a similar calculation, we have
	\begin{align}
	 \E^P\bigg[ \int_0^u \partial_t &\varphi(t+s,X_{t+s})\, ds\bigg]  \nonumber\\
	\leq & \ \int_0^u \partial_t \varphi (t,x) \,ds + \int_0^u \E^P \Big[ \big| \partial_t \varphi(t+s,X_{t+s}) - \partial_t \varphi (t,x)\big| \Big] \, ds \nonumber\\
	\leq & \ \int_0^u \partial_t \varphi (t,x) \,ds + c\int_0^u \E^P \Big[s+\sup_{0\leq v\leq u}|X_{t+v}-x| \Big]\, ds \nonumber\\
	\leq & \ \int_0^u \partial_t \varphi (t,x) \,ds + c\big(u^2 + u^{3/2}\big). \label{eq:PDE6}
	\end{align}
	As above, we write $\theta:=(b^0,b^1,a^0,a^1)$ for an element in $\Theta$.
	Then, by taking expectations in \eqref{eq:PDE1} and using \eqref{eq:PDE2}--\eqref{eq:PDE6}  yields
	\begin{align}
	\lefteqn{ \E^P\Big[\varphi (t+u,X_{t+u})- \varphi(t,x)\Big] 
	      \leq  c\big(u^3 +u^{5/2}+ u^2 + u^{3/2}\big) } \quad \nonumber\\
	      &+ \int_0^u \Big(\partial_t \varphi (t,x) + \E^P\big[
	       \partial_x \varphi(t,x) \,\beta^P_{t+s}
	       + \partial_{xx} \varphi(t,x)\, \alpha_{t+s}\big] \Big) \,ds \ \nonumber\\ 
          &\le  c\big(u^3 +u^{5/2}+ u^2 + u^{3/2}\big) + u \partial_t \varphi (t,x)\nonumber \\
          &+ \int_0^u \E^P\Big[ \sup_{\theta \in \Theta} \Big\{(b^0+b^1 X_{t+s})\,\partial_x \varphi(t,x) + \frac{1}{2} (a^0+a^1 X_{t+s}^+)\,\partial_{xx} \varphi(t,x) \Big\}\Big] .\label{eq:PDE7}
	\end{align}
	Here, the supremum turns out to be  $G(X_{t+s},\partial_x \varphi(t,x),\partial_{xx}\varphi(t,x))$. 
	Note that by the very definition of $G$,
	\begin{align*}
		G(X_{t+s},p,q) 
		& \le G(x,p,q) + \sup_{\theta \in \Theta}\Big\{ |b^1| \, |X_{t+s}-x| \, | p| +
		  |a^1| \, |X_{t+s}-x| \, | q|
		\Big\}.
	\end{align*}
	Therefore, by using that $\varphi \in C^{2,3}_b$, the definition of the constant $\mathcal K$  and Lemma~\ref{lemma 1}, we have
	\begin{align}
    	\lefteqn{\int_0^u \E^P\Big[ G(X_{t+s},\partial_x \varphi(t,x),\partial_{xx}\varphi(t,x)) \Big\}\Big] \,ds }\hspace{2cm} \nonumber \\
	    \leq & \ u  G(x,\partial_x \varphi(t,x),\partial_{xx}\varphi(t,x))
	    + u c\mathcal K \E^P\big[|X_{t+s}-x|\big] \nonumber\\
	    \leq & \ u  G(x,\partial_x \varphi(t,x),\partial_{xx}\varphi(t,x)) +  c\mathcal K \big(u^2 +u^{3/2}\big).  \label{eq:PDE8}  
	\end{align}
	Combining \eqref{eq:PDE7}--\eqref{eq:PDE8} yields
	\begin{align}
 \E^P\Big[\varphi (t+u,X_{t+u})- \varphi(t,x)\Big] 
	    &\leq  \  u \partial_t \varphi (t,x) +  u  G(x,\partial_x \varphi(t,x),\partial_{xx}\varphi(t,x))
	    \nonumber \\
	    & \ +  c\big(u^3 +u^{5/2}+ u^2 + u^{3/2}\big). \label{eq:PDE9}  
	\end{align}
	for some constant $c>0$ which is independent of $P$. As the choice of $P \in \cA(t,x,\Theta)$ was arbitrary, we deduce from
	\eqref{eq:PDE-DPP} that
	\begin{align}
	0 &\leq \sup_{P\in \cA(t,x,\Theta)}\E^P\big[\varphi(t+u,X_{t+u})-\varphi(t,x)\big]\nonumber\\
	&\leq u \partial_t \varphi (t,x) +  u  G(x,\partial_x \varphi(t,x),\partial_{xx}\varphi(t,x)) 
	 +  c\big(u^3 +u^{5/2}+ u^2 + u^{3/2}\big). \label{eq:PDE10}
	\end{align}
	By dividing first in \eqref{eq:PDE10} by $-u$ and then letting $u$ go to zero, we obtain that
	\begin{equation*} 
	- \partial_t \varphi (t,x) - G(x,\partial_x \varphi(t,x),\partial_{xx}\varphi(t,x)) \leq 0,
	\end{equation*}
	which proves that $v$ is indeed a viscosity subsolution as desired. 
\end{proof}

\end{document}